\documentclass{llncs}
\usepackage[T1]{fontenc}
\usepackage{amsmath, amssymb}
\usepackage{url}
\newcommand{\doi}[1]{\url{https://doi.org/#1}}
\usepackage{tikz}

\providecommand{\qedsymbol}{\ensuremath{\square}}

\newcommand{\qedbox}{\hfill\qedsymbol}

\title{The Path-Extremal Conjecture for Zero Forcing: Distance-Hereditary Graphs and a Split-Decomposition Reduction}

\usepackage{bbding}
\author{Samuel German \Envelope}
\institute{University of California, San Diego, USA \\
\email{sgerman@ucsd.edu}}

\begin{document}
\maketitle

\begin{abstract} A zero-forcing set of a graph is a set for which, after regarding each of its vertices as being blue, and applying the rule that a blue vertex with exactly one white neighbor forces said neighbor to become blue, all of the vertices in the graph are blue. Let $z(G;k)$ denote the number of zero forcing
sets of size $k$. A conjecture of Boyer et al.\ asserts that the path $P_n$
maximizes these numbers coefficientwise among all $n$-vertex graphs. We prove this path-extremal conjecture for
distance-hereditary graphs, which is the class of all graphs $G$ such that for each of the connected induced subgraphs of $G$, the distances between any two vertices in that subgraph are the same as they are in $G$. This extends the previously known tree case in which the path-extremal conjecture was known to be true to a much
larger class that includes, in particular, all trees and all cographs.

We then use canonical split decomposition to push the argument we used to show the distance-hereditary result into the more general setting of a connected graph with at most one prime bag in its canonical split decomposition. Specifically, we show that if a split-prime graph
$H$ and all of its induced subgraphs are path-extremal, then every connected graph whose canonical split decomposition has a
unique prime bag whose label graph is isomorphic to $H$ is also
path-extremal. As a corollary of this result, for each fixed $m$, if every induced subgraph of every
split-prime graph on at most $m$ vertices is path-extremal, then so is every
connected graph whose canonical split decomposition has a unique prime bag of
size at most $m$ (Corollary~\ref{cor:bounded-prime-size}). It is thus the case that for connected graphs having split decompositions containing at most one prime bag of bounded order that the path-extremal conjecture reduces to a finite verification problem on induced subgraphs of the possible prime-bag labels of bounded order (prime cores). Additionally, we prove that the path-extremal conjecture holds when all prime bags have order at most six,
without restricting their number; under the preceding hereditary hypothesis,
one additional prime bag labelled by $H$ is also allowed.

Taken together, our results yield a new positive instance of the
path-extremal conjecture and identify a natural structural frontier, that of canonical split decomposition,  for further
progress.
\keywords{zero forcing \and zero forcing polynomial \and distance-hereditary graph \and split decomposition \and graph-labelled tree}
\end{abstract}

\section{Introduction}

Zero forcing is a deterministic propagation process on a graph: starting from an
initial blue set $S\subseteq V(G)$, a blue vertex with exactly one white
neighbor forces that neighbor to become blue. This notion was introduced in connection with the maximum nullity and minimum rank problems for matrices described by graphs; it has indeed become an area of study drawing together concepts from both graph theory and linear algebra \cite{AIM2008,FallatHogben2007}. In line with this, the study of zero forcing contains the study of  related zero forcing parameters and minimum-rank questions, e.g. positive semidefinite zero forcing \cite{BarioliEtAl2010}. Further applications of zero forcing include propagation
time \cite{HogbenEtAl2012}, strong structural controllability
\cite{TrefoisDelvenne2015}, a priori sensor placement and error amplification
in linear recovery \cite{KenterLin2019}, and logic-circuit constructions
based on forcing gadgets \cite{BurgarthEtAl2015}. An interesting variant of zero forcing are probabilistic variants thereof such as probabilistic zero forcing
\cite{KangYi2013}, which are considered in work on random graphs and expected
propagation time \cite{EnglishMacRuryPralat2021}.

It is natural to count forcing sets by
their size.  For an $n$-vertex graph $G$, let $z(G;k)$ denote the number of
zero forcing sets of size $k$, and consider the zero forcing polynomial
$\mathcal{Z}(G;x)=\sum_{k=1}^{n} z(G;k)x^k$. This polynomial was introduced by Boyer et al., who also laid forth the systematic study of counting zero-forcing sets, among other things establishing structural properties of
$\mathcal{Z}(G;x)$ such as extremal coefficients, multiplicativity,
unimodality questions, and recognizability phenomena \cite{BoyerEtAl2019}. Strikingly, Boyer et al.\ conjectured that the path $P_n$ maximizes the number of zero-forcing sets of every size out of all graphs on $n$ vertices; that is, that $P_n$ coefficient-wise maximizes the zero forcing polynomial \cite{BoyerEtAl2019}.

Significant progress has recently been made in the study of this path-extremal conjecture. Curtis et al. showed that it holds in the case of $n$-vertex graphs with minimum degree at least $\log_2 n+2\log_2\log_2 n$ \cite{CurtisEtAl2024}; they also showed that for every sufficiently large $n$,
every $n$-vertex tree $T$, and every $p\in[0,1]$,
the probability that a random subset of $V(T)$, formed by
including each vertex independently with probability $p$,
is zero forcing is at most the corresponding probability
for $P_n$ \cite{CurtisEtAl2024}. Another significant result in the realm of the path-extremal conjecture is that Using operations on graph such as edge-deletion, Menon and Singh were able to show that this conjecture holds for outerplanar graphs and
threshold graphs, a corollary of which is that the path-extremal conjecture holds in the case of trees \cite{MenonSingh2025}. One of the lessons established by their work is that the conjecture seems to be able to be dealt with on classes of graphs having a nice recursive structure. Our paper continues this same type of decomposition-driven approach.

Our first main result proves the path-extremal conjecture for
distance-hereditary graphs (Theorem~\ref{thm:dh-path-extremal}). Distance-hereditary graphs are a natural next target on which to prove the path-extremal conjecture for several reasons: they form a large a classical graph class, admitting a one-vertex extension
characterization via pendant, false-twin, and true-twin additions
\cite{BandeltMulder1986}, and, equivalently, they are exactly the totally
decomposable graphs for split decomposition, namely those whose canonical
graph-labelled tree has only clique and star bags \cite{GioanPaul2012}; these structural characterizations match the two basic counting mechanisms that now
underpin many positive results on the conjecture in that twin pairs produce small forts
and hence many non-forcing sets, while leaves give recursive lower bounds on the
number of non-forcing sets \cite{BoyerEtAl2019,MenonSingh2025}; in this sense,
distance-hereditary graphs are the largest standard class for which the
``twin/leaf'' strategy is expected to work without any prime bags, as the canonical-split decompositions of distance-hereditary graphs contain none of them.

Our second main result pushes the argument we used to show the distance-hereditary result into the more general setting of a connected graph with at most one prime bag in its canonical split decomposition. Namely our Theorem~\ref{thm:one-prime-bag} states that if a fixed split-prime graph $H$ and all of
its induced subgraphs are path-extremal, then every connected graph whose canonical split decomposition has a unique prime bag
whose label graph is isomorphic to $H$ is also path-extremal. A corollary of this is that for each fixed $m$, if every induced subgraph of every
split-prime graph on at most $m$ vertices is path-extremal, then so is every
connected graph whose canonical split decomposition has a unique prime bag of
size at most $m$ (Corollary~\ref{cor:bounded-prime-size}). This reduces the truth of the path-extremal conjecture on an infinite class of graphs, the class of graphs whose canonical split-decomposition has a unique prime bag of bounded size, to finitely many checks on bounded-order split-prime cores, using the concept of split decomposition
\cite{GioanPaul2012}.

The first way in which these results are significant is that a new, substantially large class of graphs have the path-extremal being proven true for them, which extends upon the previously known case of trees to a class much larger than it and which in particular contains all trees, cographs, and many graphs built by repeated pendant and twin extensions
\cite{BandeltMulder1986,GioanPaul2012,MenonSingh2025}. It is also significant that Theorem~\ref{thm:one-prime-bag} identifies a structural frontier at which the structural idea we employ to show the path-extremal conjecture in the case of distance-hereditary graphs fails to work in the general case while also showing how far it can be pushed once we stipulate that there is at most one prime bag. Third,
Corollary~\ref{cor:bounded-prime-size} suggests a concrete direction for
future work: one can combine induction on split decomposition with finite
verification of bounded prime cores. Our reduction also has an algorithmic appeal in that one can compute whether a graph is part of the class of graphs for which our reduction holds, namely the class of graphs whose canonical split-decomposition has a unique prime bag of size at most $m$, in polynomial-time \cite{GioanPaul2012}.

Theorem~\ref{thm:small-prime-bags} further removes the restriction on the
number of prime bags of order at most six, even in the presence of one
additional prime bag satisfying the hereditary hypothesis above.

\section{Definitions and Preliminaries}

All graphs in this paper are finite, simple, and undirected. We use $V(G)$ to denote the vertex set of a graph $G$, and we use $E(G)$ to denote the same graph's edge set. We use $n$ in place of $|V(G)|$. Furthermore, for a vertex $v \in V(G)$, the open neighborhood $N_G(v)$ of $v$ is the set of all vertices to which it is adjacent in $G$, so $N_G(v) := \{u\in V(G): uv\in E(G)\}$ (which given our condition of a simple class excludes $v$), and the closed neighborhood of $v$ is the open neighborhood of $v$ with the addition of $v$ itself, so $N_G[v] := N_G(v) \cup \{v\}$. Thus $\deg_G(v) := |N_G(v)|$. Whenever we are working in a context in which some specific graph is salient, we drop the subscripts in these, thus writing 
$N(v)$ and $\deg(v)$, with the understanding that these apply to the salient graph in the context. For a set $X \subseteq V(G)$, the induced subgraph on $X$ is
denoted by $G[X]$, and $G-X$ denotes the induced subgraph $G[V(G)\setminus X]$.
For vertices $u,v \in V(G)$, we write $u \sim v$ if $uv \in E(G)$ and
$u \nsim v$ otherwise. The path on $n$ vertices is denoted by $P_n$.
A vertex of degree~$1$ is called a leaf or pendant vertex.

\subsection{Zero Forcing}

Let $G$ be a graph and let $S \subseteq V(G)$. We regard the vertices of $S$ as
initially blue and the vertices of $V(G)\setminus S$ as initially white.
The \emph{zero forcing color-change rule} is the following:

\noindent
\emph{If a blue vertex has exactly one white neighbor, then it forces that neighbor
to become blue.}

Starting from $S$, we repeatedly apply the color-change rule until no further force
is possible. The set of vertices that are blue at the end of this process is denoted
by $\operatorname{cl}_G(S)$ and is called the \emph{zero forcing closure} of $S$ in $G$.

A set $S \subseteq V(G)$ is a \emph{zero forcing set} of $G$ if
$\operatorname{cl}_G(S) = V(G)$. The \emph{zero forcing number} of $G$ is
$Z(G) := \min\{|S| : S \subseteq V(G)\text{ is a zero forcing set of }G\}$.

For $0 \le k \le n$, let \[z(G;k) := \bigl|\{S \subseteq V(G) : |S|=k \text{ and } S \text{ is a zero forcing set of }G\}\bigr|\]
denote the number of zero forcing sets of size $k$. We also write $z'(G;k) := \binom{n}{k} - z(G;k)$
for the number of \emph{non-forcing} $k$-subsets of $V(G)$.

\noindent\textbf{Convention.}
For notational convenience, we extend the definitions of $z(G;k)$ and $z'(G;k)$
to all integers $k$ by setting
$z(G;k)=z'(G;k)=0$ whenever $k<0 \text{ or } k> |V(G)|$.
Likewise, throughout the paper we use the standard convention $\binom{a}{b}=0$ whenever $b<0 \text{ or } b>a$.

The \emph{zero forcing polynomial} of $G$ is
$\mathcal{Z}(G;x) := \sum_{k=1}^{n} z(G;k)x^k$. 

For two polynomials
$f(x)=\sum_{k \ge 0} a_k x^k$ and
$g(x)=\sum_{k \ge 0} b_k x^k$,
we write $f(x) \preceq g(x)$
if $a_k \le b_k$ for every $k$.
Thus $\mathcal{Z}(G;x) \preceq \mathcal{Z}(H;x)$ means
$z(G;k) \le z(H;k)$ for all $k$.

When $|V(G)|=n$, we say that $G$ is \emph{path-extremal} if
$z(G;k) \le z(P_n;k)$ for all $0 \le k \le n$.
Equivalently, $G$ is path-extremal if $z'(G;k) \ge z'(P_n;k)$ for all $0 \le k \le n$.

\subsection{Forts and Twins}

A nonempty set $F \subseteq V(G)$ is a \emph{fort} if every vertex in $V(G)\setminus F$ has
either $0$ neighbors in $F$ or at least $2$ neighbors in $F$. Their significance is that they serve as obstructions to zero forcing, in that if an initial blue set has empty
intersection with a fort, then the zero forcing process cannot enter that fort.

Two distinct vertices $u,v \in V(G)$ are \emph{true twins} if
$u \sim v $
and
$N(u)\setminus\{v\} = N(v)\setminus\{u\}$,
and they are \emph{false twins} if
$u \nsim v$ and $N(u)=N(v)$.
A pair $\{u,v\}$ is called a \emph{twin pair} if $u$ and $v$ are either true twins
or false twins. This is equivalently said as a two-element set $\{u,v\}$ is a fort if and
only if $\{u,v\}$ is a twin pair, which is because every vertex outside the pair is adjacent to
both or to neither, while $u$ and $v$ may be adjacent or nonadjacent.

\subsection{Distance-Hereditary Constructions}

Let $G$ be a graph and let $v \in V(G)$.

\begin{itemize}
    \item A \emph{pendant addition at $v$} is the operation of adding a new vertex
    $u \notin V(G)$ adjacent to exactly one vertex, namely $v$.

    \item A \emph{false-twin addition at $v$} is the operation of adding a new
    vertex $u \notin V(G)$ such that $u \nsim v$ and $N_{G'}(u)=N_G(v)$,
    where $G'$ denotes the resulting graph.

    \item A \emph{true-twin addition at $v$} is the operation of adding a new
    vertex $u \notin V(G)$ such that $u \sim v$ and
    $N_{G'}(u)\setminus\{v\}=N_G(v)$.
\end{itemize}

A graph is \emph{distance-hereditary} if it can be obtained from the one-vertex graph
$K_1$ by a sequence of pendant additions, false-twin additions, and true-twin additions.

\subsection{Splits and Split Decomposition}

From this point onward, whenever split decomposition is discussed, all graphs are
assumed to be connected.

A \emph{split} of a connected graph $G$ is a partition $(A,B)$ of $V(G)$ such that
$|A|,|B| \ge 2$ and there exist nonempty sets $A_1 \subseteq A$ and $B_1 \subseteq B$
for which
\[
ab \in E(G)
\quad\Longleftrightarrow\quad
a \in A_1 \text{ and } b \in B_1
\]
for every $a \in A$ and $b \in B$. The displayed condition, for $A, B, A_1, B_1$, just says that all of the edges crossing the cut between $A$ and $B$ constitute exactly the edges of the complete bipartite graph $K_{A_1, B_1}$.
A connected graph is \emph{split-prime} if it admits no split and is neither a
clique nor a star. Equivalently, in the terminology of canonical split
decomposition, split-prime graphs are the nondegenerate no-split graphs that
occur as prime bag labels \cite{GioanPaul2012}.

We use the standard graph-labelled-tree model of the \emph{canonical split decomposition};
see, for example, \cite{GioanPaul2012}. Its internal nodes are called \emph{bags}.
Each bag carries a label graph, and each internal edge of the decomposition tree
corresponds to a distinguished vertex of the label graph, called a \emph{marker vertex}.
Our terminology is a compressed version of that in
\cite[Definitions~2.1--2.2]{GioanPaul2012}: their internal tree nodes are our
bags, and their leaves of the graph-labelled tree are the vertices of $G$. We
suppress the pendant leaf edges and identify each such leaf with the corresponding
vertex of the adjacent label graph; these identified label vertices are called
\emph{ordinary vertices}, while marker vertices are the label vertices corresponding
to internal tree edges.
The label of each bag is one of the following:
\begin{itemize}
    \item a \emph{clique bag}, whose label graph is a clique;
    \item a \emph{star bag}, whose label graph is a star;
    \item a \emph{prime bag}, whose label graph is split-prime in the above
    nondegenerate sense.
\end{itemize}
Thus clique bags and star bags account for the degenerate cases, and a
\emph{prime bag} always means a nondegenerate no-split bag.

A bag is a \emph{leaf bag} if it is a leaf of the tree obtained from the
full graph-labelled tree by retaining only the bags and the edges between them.
If $B$ is a leaf bag, then exactly one vertex of its label graph is a marker
vertex; we denote this vertex by $m_B$ and call it the \emph{marker} of $B$.
All other vertices of the label graph of $B$ are called the \emph{ordinary vertices}
of $B$.

If $B$ is a star leaf bag, then its label graph has a unique center.
We say that $B$ is \emph{center-attached} if $m_B$ is the center of that star, and
\emph{leaf-attached} otherwise.

Each vertex of a bag label is either an ordinary vertex
or a marker vertex. What the ordinary vertices are across all
bag labels is the vertices of $G$, with each vertex
of $G$ appearing in exactly one bag. Each edge of the
decomposition tree is represented by a pair of marker vertices,
one in each incident bag; these are auxiliary vertices and do
not belong to $V(G)$. Adjacency in $G$ is recovered through the
standard \emph{accessibility} relation on the graph-labelled
tree: viz., two ordinary vertices belonging to the
same bag are adjacent in $G$ if and only if they are adjacent
in the corresponding label graph.
Otherwise, let $u$ and $v$ be ordinary vertices lying in bags
$B_0$ and $B_t$, respectively, and let
$B_0,B_1,\dots,B_t$
be the unique path between these bags in the decomposition tree. For
$0\le i<t$, let $m_i^+$ be the marker vertex of $B_i$ corresponding to the
tree edge $B_iB_{i+1}$, and for $1\le i\le t$, let $m_i^-$ be the marker
vertex of $B_i$ corresponding to the tree edge $B_{i-1}B_i$. Then $u$ and
$v$ are adjacent in $G$ if and only if $u$ is adjacent to $m_0^+$ in the
label of $B_0$, $m_i^-$ is adjacent to $m_i^+$ in the label of $B_i$ for
every $1\le i\le t-1$, and $m_t^-$ is adjacent to $v$ in the label of
$B_t$. We refer to this as the accessibility relation.

A graph-labelled-tree representation of a connected graph will be called
\emph{reduced} in the standard split-decomposition sense of every clique bag and
every star bag having at least three vertices in its label graph, and no internal
edge admitting one of the standard degenerate-bag joins, namely a clique-clique join
or a star leaf-marker to star center-marker join (see
\cite[Section~2.2, definition before Theorem~2.9]{GioanPaul2012}).

\begin{lemma}\label{lem:reduced-basic-facts}
Let $T$ be a reduced graph-labelled-tree representation of a connected graph
$G$.
\begin{enumerate}
    \item Every clique bag and every star bag of $T$ has at least three
    vertices in its label graph.
    \item If $B$ is a clique leaf bag of $T$, then $B$ has at least two
    ordinary vertices.
    \item If $B$ is a center-attached star leaf bag of $T$, then $B$ has
    at least two ordinary leaves.
  
\end{enumerate}
\end{lemma}

\begin{proof}
Assertion (1) is the degree greater than two condition in the definition of
reduced graph-labelled trees in \cite[Section~2.2, definition before
Theorem~2.9]{GioanPaul2012}. Since a leaf bag has exactly one marker vertex,
assertions (2) and (3) follow immediately from (1).  \qedbox
\end{proof}

\begin{lemma}[Reducedness constraints on internal edges]
\label{lem:reduced-edge-types}
Let $T$ be a reduced graph-labelled-tree representation of a connected graph.
Then:
\begin{enumerate}
    \item no internal tree edge of $T$ joins two clique bags; and
    \item no internal tree edge of $T$ joins a star bag through a leaf-marker
    to a star bag through a center-marker.
\end{enumerate}

\end{lemma}

This is nothing but the formulation of the reducedness condition in
Cunningham's split-decomposition theorem \cite{Cunningham1982}; in the terminology
used here, that condition is stated by Gioan and Paul as being that reducedness excludes
clique-joins and star-joins: two clique bags cannot be adjacent,
and two star bags cannot be adjacent through the center-marker
of one and a leaf-marker of the other
\cite[Section~2.2, definition before Theorem~2.9]{GioanPaul2012}.

We say that a connected graph has a \emph{unique prime bag} if its canonical split decomposition
contains exactly one prime bag.

We classify connected graphs with exactly one prime bag
in their canonical split decomposition according to the
label of that bag. Given a split-prime graph $H$, let
$\mathcal{C}(H)$ consist of all such graphs for which
this label is isomorphic to $H$. We call the unique
prime bag the \emph{prime core} of the graph.

For $m \in \mathbb{N}$, we write $\mathcal{C}_{\le m}$ for the class of connected
graphs whose canonical split decomposition has exactly one prime bag of size at most $m$.
Here the \emph{size} of a bag means the number of vertices in its label graph.

Finally, a connected graph is distance-hereditary if and only if its canonical split
decomposition has no prime bag.

\section{Distance-Hereditary Graphs Are Path-Extremal}

This section is devoted towards proving the the path-extremal conjecture for distance-hereditary
graphs. We start this with two ingredients from the recent literature: an explicit formula for
the coefficients of the zero forcing polynomial of a path, and a recurrence for
non-forcing sets in the presence of a leaf.

\begin{lemma}[\cite{BoyerEtAl2019}]\label{lem:path-nonforcing}
For every $n \ge 1$ and every $0 \le k \le n$,
$z'(P_n;k)=\binom{n-k-1}{k}$.
Equivalently,
$z(P_n;k)=\binom{n}{k}-\binom{n-k-1}{k}$.
\end{lemma}

\begin{lemma}[\cite{MenonSingh2025}]\label{lem:leaf-recursion}
Let $G$ be a graph with a leaf $x$, and let $v$ be the unique neighbor
of $x$. Then, for every $k \ge 1$,
$z'(G;k)\ge z'(G-x;k)+z'(G-\{x,v\};k-1)$.
\end{lemma}

We also use the standard hereditary property of distance-hereditary graphs.

\begin{lemma}[\cite{BandeltMulder1986}]\label{lem:dh-hereditary}
Every induced subgraph of a distance-hereditary graph is distance-hereditary.
\end{lemma}

\begin{lemma}\label{lem:fort-obstruction}
Let $F \subseteq V(G)$ be a fort. If $S\cap F=\varnothing$, then
$S$ is not a zero forcing set of $G$.
\end{lemma}

\begin{proof}
Suppose that $S$ is a zero forcing set. Consider the zero forcing process
started from $S$, and let $w\notin F$ be the first blue vertex that forces a
vertex of $F$. Before its inception $F$ consists of only white vertices. $F$ being a fort implies that $w$ has either $0$ or at least $2$ neighbors which are in $F$; but the first of which cases cannot happen because $w$ is forcing a vertex of $F$; and the second case contradicts the color-change rule as it necessitates $w$ having at least two white neighbors; it follows that no force can enter $F$, whence $S$ is not a zero forcing set. \qedbox
\end{proof}

\begin{lemma}\label{lem:twin-pair-fort}
If $u$ and $v$ form a twin pair in $G$, then $\{u,v\}$ is a fort.
\end{lemma}

\begin{proof}
Let $w\in V(G)\setminus\{u,v\}$. Since $u$ and $v$ are twins, we have
$w \sim u \iff w \sim v$.
Thus $w$ has either $0$ or $2$ neighbors in $\{u,v\}$, and so
$\{u,v\}$ is a fort. \qedbox
\end{proof}

\begin{corollary}\label{cor:twin-path-extremal}
If $G$ contains a twin pair, then $G$ is path-extremal.
\end{corollary}

\begin{proof}
Let $\{u,v\}$ be a twin pair in $G$. By Lemma~\ref{lem:twin-pair-fort},
the set $\{u,v\}$ is a fort. Hence, by Lemma~\ref{lem:fort-obstruction},
every $k$-subset of $V(G)\setminus\{u,v\}$ is non-forcing. Therefore,
if $n=|V(G)|$, then
$z'(G;k)\ge \binom{n-2}{k}$
for every $0\le k\le n$. For $k=0$, this already gives
$z'(G;0)\ge z'(P_n;0)$. For $k\ge 1$, we have $n-2\ge n-k-1$, and hence
$\binom{n-2}{k}\ge \binom{n-k-1}{k}=z'(P_n;k)$
by Lemma~\ref{lem:path-nonforcing}. Thus
$z'(G;k)\ge z'(P_n;k)$
for all  $0\le k\le n$.
Equivalently, $G$ is path-extremal. \qedbox
\end{proof}

\begin{lemma}\label{lem:dh-leaf-or-twin}
Let $G$ be a distance-hereditary graph with at least two vertices.
Then $G$ contains either a leaf or a twin pair.
\end{lemma}

\begin{proof}
By definition, $G$ can be obtained from $K_1$ by a sequence of pendant
additions, false-twin additions, and true-twin additions. Consider the  last such operation when constructing $G$. It being a pendant addition implies that the vertex it adds is a leaf of $G$; it being a false or true-twin addition implies that the new vertex constitutes a twin pair with respect to the vertex to which it was added; and the desired statement follows in either case. 
 \qedbox
\end{proof}

We now prove the main result of this section.

\begin{theorem}\label{thm:dh-path-extremal}
Every distance-hereditary graph is path-extremal. Equivalently, if $G$ is a
distance-hereditary graph on $n$ vertices, then $\mathcal{Z}(G;x)\preceq \mathcal{Z}(P_n;x)$.
\end{theorem}

\begin{proof}
We argue by induction on $n=|V(G)|$. The statement is immediate for $n\le 2$.

Now let $n\ge 3$, and let $G$ be a distance-hereditary graph on $n$
vertices. By Lemma~\ref{lem:dh-leaf-or-twin}, the graph $G$ contains either a
twin pair or a leaf.

If $G$ contains a twin pair, then the result follows immediately from
Corollary~\ref{cor:twin-path-extremal}.

Assume instead that $G$ contains a leaf $x$, and let $v$ be the unique
neighbor of $x$. By Lemma~\ref{lem:dh-hereditary}, both $G-x$ and
$G-\{x,v\}$ are distance-hereditary. Since these graphs have fewer than $n$
vertices, the induction hypothesis gives
$z'(G-x;k)\ge z'(P_{n-1};k)=\binom{n-k-2}{k}$.
for all $0\le k\le n-1$, and $z'(G-\{x,v\};k-1)\ge z'(P_{n-2};k-1)=\binom{n-k-2}{k-1}$
for all $1\le k\le n$, where we used
Lemma~\ref{lem:path-nonforcing} in both equalities.

Applying Lemma~\ref{lem:leaf-recursion}, we obtain, for every $1\le k\le n$,
$z'(G;k)
   \ge z'(G-x;k)+z'(G-\{x,v\};k-1)
   \ge \binom{n-k-2}{k}+\binom{n-k-2}{k-1}$.
By Pascal's identity,
$\binom{n-k-2}{k}+\binom{n-k-2}{k-1}
   = \binom{n-k-1}{k}
   = z'(P_n;k)$,
again by Lemma~\ref{lem:path-nonforcing}. Thus $z'(G;k)\ge z'(P_n;k)$ for all  $1\le k\le n$.
The case $k=0$ is trivial, so $G$ is path-extremal. \qedbox
\end{proof}

\begin{corollary}
Every tree and every cograph is path-extremal.
\end{corollary}

\begin{proof}
Both classes are distance-hereditary, whence the desired claim follows immediately. \qedbox
\end{proof}

\section{Extensions via Prime Bags}

We now extend Theorem~\ref{thm:dh-path-extremal} by allowing a single prime bag
in the split decomposition. The proof proceeds by induction on the number of vertices, using reductions that delete pendant vertices and, in some cases, their neighbors. To formulate the induction more generally, we enlarge $\mathcal{C}(H)$ to include graphs with no prime bag, as well as graphs whose unique prime-bag label is isomorphic to an induced subgraph of $H$.

For a fixed split-prime graph $H$, let $\mathcal{D}(H)$ denote the class of
connected graphs that admit a reduced graph-labelled-tree representation with at
most one prime bag, and such that the label graph of the prime bag, if it
exists, is isomorphic to an induced subgraph of $H$. In particular,
$\mathcal{C}(H)\subseteq \mathcal{D}(H)$.
Moreover, if a graph in $\mathcal{D}(H)$ has no prime bag, then it is totally
decomposable and hence distance-hereditary \cite{GioanPaul2012}.

\begin{lemma}\label{lem:leaf-bag-geometry}
Let $T$ be a reduced graph-labelled-tree representation of a connected graph
$G$, and let $B$ be a non-prime leaf bag of $T$.
\begin{enumerate}
    \item If $B$ is a clique bag, then $G$ contains a twin pair.
    \item If $B$ is a center-attached star bag, then $G$ contains a twin pair.
    \item If $B$ is a leaf-attached star bag with center $c$, then every
    ordinary leaf of $B$ is a leaf of $G$ adjacent to $c$.
\end{enumerate}
\end{lemma}

\begin{proof}
Let $m_B$ be the marker of $B$. 

Assume first that $B$ is a clique bag. By Lemma~\ref{lem:reduced-basic-facts}, the bag $B$ has at least two ordinary vertices. Now for any ordinary vertex of $B$, it is adjacent to $m_B$ inside said clique; this implies that all ordinary vertices of $B$ see the rest of
the graph through the same marker so it has the same neighborhood
outside $B$. All pairs of ordinary vertices of $B$ are adjacent and thus any of them of them constitute true twins in $G$.

For the proof of the second enumerated if statement, assume its antecedent, i.e.  assume that $B$ is a center-attached star bag. It is again the case that Lemma~\ref{lem:reduced-basic-facts} implies that $B$ has at least two ordinary leaves. Each ordinary leaf is adjacent to
the marker $m_B$, which is now the center of the star, and no two ordinary
leaves are adjacent to one another, whence it follows that any two ordinary leaves have the same
neighborhood outside $B$ and are nonadjacent, so they form a pair of false
twins in $G$.

 For the final if statement enumerated, assume that $B$ is a leaf-attached star bag with center $c$. Let
$x$ be an ordinary leaf of $B$. Inside $B$, the vertex $x$ is adjacent
only to $c$, and it is not adjacent to the marker $m_B$; as such, $x$
has no neighbor outside $B$. Therefore $x$ has degree $1$ in $G$, and
its unique neighbor is $c$. \qedbox
\end{proof}

\begin{corollary}\label{cor:multi-leaf-star-twins}
Let $T$ be a reduced graph-labelled-tree representation of a connected graph
$G$, and let $B$ be a leaf-attached star bag of $T$. If $B$ has at
least two ordinary leaves, then $G$ contains a twin pair.
\end{corollary}

\begin{proof}
Let $c$ be the center of $B$. By Lemma~\ref{lem:leaf-bag-geometry}, every
ordinary leaf of $B$ is a leaf of $G$ adjacent to $c$. Hence any two
ordinary leaves of $B$ are nonadjacent and have the same unique neighbor
$c$, so they are false twins. \qedbox
\end{proof}

The below lemma identifies exactly when one can peel a non-prime leaf bag away
from the prime bag.

\begin{lemma}\label{lem:leaf-far-from-prime}
Let $T$ be a reduced graph-labelled-tree representation of a connected graph
$G$, and suppose that $T$ has a unique prime bag $P$. If $T$ is not a
star centered at $P$, then there exists a non-prime leaf bag $B$ whose
unique neighbor bag is also non-prime.
\end{lemma}

\begin{proof} We know that there is some bag which is distance at least $2$ from $P$, due to how $T$ is not a star centered at $P$, there exists a bag at distance at
least $2$ from $P$. Now let $B$ be a bag such that $B\neq P$ and 
$\operatorname{dist}_T(B,P)$ is maximized. $B$ has to be a leaf bag as every vertex farthest from a fixed vertex is necessarily a leaf in a tree. Finally, as $\operatorname{dist}_T(B,P)\ge 2$, $P$ is not the unique neighbor of $B$ and the neighbor bag of $B$ must be non-prime as $P$ is the only prime bag.  \qedbox
\end{proof}

The function of the next lemma is to allow us to conclude that in the case of not being a star, removing a leaf-attached star
bag with a single ordinary leaf preserves membership in $\mathcal{D}(H)$.

\begin{lemma}\label{lem:peeling-preserves-DHclass}
Let $H$ be a split-prime graph, let $G\in \mathcal{D}(H)$, and fix a reduced
graph-labelled-tree representation $T$ of $G$ with unique prime bag $P$.
Suppose that $T$ is not a star centered at $P$, and let $B$ be a non-prime
leaf bag whose unique neighbor bag is non-prime. Assume that $B$ is a
leaf-attached star bag with exactly one ordinary leaf $x$ and center $c$.
Let $A$ be the unique neighbor bag of $B$, and let $a$ be the marker
vertex of $A$ corresponding to the tree edge joining $A$ to $B$. Then:
\begin{enumerate}
    \item $G-x\in \mathcal{D}(H)$;
    \item $G-\{x,c\}\in \mathcal{D}(H)$.
\end{enumerate}
\end{lemma}

   \begin{proof}[Proof sketch]
Let $m_B$ be the marker of $B$. Since $B$ is a leaf-attached star bag
with exactly one ordinary leaf $x$, its label graph consists of the three
vertices $m_B,c,x$, with edges $m_Bc$ and $cx$.

Because $B$ is leaf-attached, the marker $m_B$ is a leaf of the star label
of $B$. Hence, by Lemma~\ref{lem:reduced-edge-types}, if the neighboring bag
$A$ is a star bag, then the corresponding marker $a$ in $A$ cannot be the
center of that star; thus $a$ is a leaf whenever $A$ is a star bag.

For (1), delete the bag $B$ and the tree edge $AB$, and reinterpret the
marker $a$ of $A$ as an ordinary vertex labelled $c$. By the accessibility
relation, this changes neither the adjacencies among the remaining old ordinary
vertices nor the neighborhood of $c$, so the resulting graph-labelled tree
represents $G-x$. No prime bag label changes, and after the standard reduction
operations if necessary, we obtain a reduced graph-labelled-tree representation
of $G-x$ with the same unique prime bag label as before. Hence
$G-x\in\mathcal{D}(H)$.

For (2), instead delete the bag $B$, the tree edge $AB$, and the marker
$a$ from the label graph of $A$. Again by the accessibility relation, the
resulting graph-labelled tree represents $G-\{x,c\}$, since any accessibility
path using the edge $AB$ necessarily ended at one of the deleted vertices
$x$ or $c$. Because $A$ is non-prime and $a$ is either a clique vertex
or a leaf of a star, deleting $a$ leaves a clique bag or a star bag; after
standard reductions if necessary, the unique prime bag label is preserved.
Therefore $G-\{x,c\}\in\mathcal{D}(H)$.

\qedbox
\end{proof}

\begin{lemma}\label{lem:prime-star-reduction}
Let $H$ be a split-prime graph, let $G\in \mathcal{D}(H)$, and fix a reduced
graph-labelled-tree representation $T$ of $G$ whose unique prime bag is
$P$. Suppose that $T$ is a star centered at $P$. Then one of the
following holds:
\begin{enumerate}
    \item $G$ contains a twin pair; or
    \item there exists an induced subgraph $Q$ of $G$ isomorphic to the
    label graph of $P$ such that $G$ is obtained from $Q$ by attaching
    pendant vertices to vertices of $Q$.
\end{enumerate}
\end{lemma}

\begin{proof}

We assume that every non-prime leaf bag is a leaf-attached star
bag with exactly one ordinary leaf, which is valid as if $B$ is any non-prime leaf bag adjacent to $P$, $B$ being either a clique bag or center-attached star bag implies it contains a twin pair by
Lemma~\ref{lem:leaf-bag-geometry}, and Corollary~\ref{cor:multi-leaf-star-twins} gives us the containment of a twin pair if $B$ is a leaf-attached star bag with at
least two ordinary leaves.  Let the leaf bags be $B_1,\dots,B_t$. For
each $i$, let $m_i$ be the marker of $B_i$, let $c_i$ be the center of
$B_i$, let $x_i$ be the unique ordinary leaf of $B_i$, and let $p_i$ be
the marker vertex of $P$ corresponding to the tree edge joining $P$ to
$B_i$.

Let $O(P)$ denote the set of ordinary vertices of the prime bag $P$, and
define $Q := G[\,O(P)\cup \{c_1,\dots,c_t\}\,]$.
We claim that $Q$ is isomorphic to the label graph of $P$. Let
$\phi : V(Q) \to V(P)$ be a map which fixes each ordinary vertex of $P$ and which maps $c_i$ to $p_i$.
f
We now verify the adjacency conditions. If $u,v\in O(P)$, then $u$ and $v$ are adjacent in $Q$
if and only if they are adjacent in the label graph of $P$, since they lie in
the same bag. If $u\in O(P)$, then $c_i$ is adjacent to $u$ in $G$ if
and only if $c_i$ is adjacent to $m_i$ in $B_i$ and $p_i$ is adjacent
to $u$ in $P$, by the accessibility relation. Since $c_i\sim m_i$ in
$B_i$, this is equivalent to $p_i\sim u$ in $P$. Similarly, for $i\neq j$,
the vertices $c_i$ and $c_j$ are adjacent in $G$ if and only if
$c_i\sim m_i$ in $B_i$, $p_i\sim p_j$ in $P$, and $m_j\sim c_j$ in
$B_j$; again this is equivalent to $p_i\sim p_j$ in $P$. Thus $\phi$ is
an isomorphism from $Q$ to the label graph of $P$.

Finally, by Lemma~\ref{lem:leaf-bag-geometry}, each $x_i$ is a leaf of $G$
adjacent only to $c_i$. Hence $G$ is obtained from $Q$ by attaching the
pendant vertices $x_1,\dots,x_t$ to the vertices $c_1,\dots,c_t$,
respectively. \qedbox
\end{proof}

\begin{lemma}\label{lem:pendant-extension}
Let $Q$ be a graph such that every induced subgraph of $Q$ is path-extremal.
Let $G$ be a graph obtained from $Q$ by attaching pendant vertices to
vertices of $Q$. Then $G$ is path-extremal.
\end{lemma}

\begin{proof}
We proceed by induction on $t$ for $t$ being the number of pendant vertices which are attached to $Q$. The base case of $t = 0$ is trivial, so we now do the inductive hypothesis.

Assume $t\ge 1$. If some vertex of $Q$ supports at least two added pendant
vertices, then those pendant vertices are false twins in $G$, and so $G$ is
path-extremal by Corollary~\ref{cor:twin-path-extremal}; thus we have license to assume that there doesn't exist a vertex of $Q$ supporting multiple pendant vertices, which we now do.  Let $x$ be one of the pendant vertices, and let $v\in V(Q)$
be its unique neighbor.

Then $G-x$ is obtained from $Q$ by attaching $t-1$ pendant vertices, so
$G-x$ is path-extremal by the induction hypothesis. Moreover, because $v$
supports no other pendant vertex, the graph $G-\{x,v\}$ is obtained from the
induced subgraph $Q-v$ by attaching at most $t-1$ pendant vertices. Since
every induced subgraph of $Q$ is path-extremal, the same is true of $Q-v$,
and hence the induction hypothesis implies that $G-\{x,v\}$ is path-extremal.

Applying Lemma~\ref{lem:leaf-recursion} to the leaf $x$, we obtain for every
$k\ge 1$,
$z'(G;k)\ge z'(G-x;k)+z'(G-\{x,v\};k-1)
        \ge z'(P_{n-1};k)+z'(P_{n-2};k-1)
        = z'(P_n;k)$,
where $n=|V(G)|$, and the last equality follows from
Lemma~\ref{lem:path-nonforcing} and Pascal's identity. $k=0$ is
trivial so $G$ is path-extremal. \qedbox
\end{proof}

\begin{theorem}\label{thm:one-prime-bag}
Let $H$ be a split-prime graph. Suppose that every induced subgraph of $H$
is path-extremal. Then every graph in $\mathcal{C}(H)$ is path-extremal.
\end{theorem}

\begin{proof}
We prove the stronger statement that every graph in $\mathcal{D}(H)$ is
path-extremal, by induction on $n=|V(G)|$.

Let $G\in \mathcal{D}(H)$, and fix a reduced graph-labelled-tree
representation $T$ of $G$ with at most one prime bag.

If $T$ has no prime bag, then $G$ is distance-hereditary, and the result
follows from Theorem~\ref{thm:dh-path-extremal}. Thus we may assume that $T$
has a unique prime bag $P$.

If $T$ has only one bag, then $G$ is isomorphic to the label graph of $P$.
From the label graph of $P$ being isomorphic to an induced subgraph of $H$ it follows that the hypothesis on $H$ implies that $G$ is path-extremal.

We now assume that $T$ has more than one bag.

\noindent\emph{Case 1: $T$ is not a star centered at $P$.}
By Lemma~\ref{lem:leaf-far-from-prime}, there exists a non-prime leaf bag $B$
whose unique neighbor bag is non-prime.

$G$ is path-extremal
by Corollary~\ref{cor:twin-path-extremal} if either $B$ is a clique bag or a center-attached star bag: $G$ must contain a twin pair here, by Lemma~\ref{lem:leaf-bag-geometry}. If $B$ is a leaf-attached star bag
with at least two ordinary leaves, then $G$ contains a twin pair by
Corollary~\ref{cor:multi-leaf-star-twins}, and again $G$ is path-extremal.

Thus it is valid to make the assumption of $B$ being  a leaf-attached star bag with exactly one
ordinary leaf $x$, so we make this assumption and let $c$ be the center of $B$. By
Lemma~\ref{lem:leaf-bag-geometry}, the vertex $x$ is a leaf of $G$ with
unique neighbor $c$. Hence Lemma~\ref{lem:leaf-recursion} gives $z'(G;k)\ge z'(G-x;k)+z'(G-\{x,c\};k-1)$
for every $k\ge 1$.

By Lemma~\ref{lem:peeling-preserves-DHclass}(1), we have
$G-x\in \mathcal{D}(H)$, so the induction hypothesis implies that $G-x$ is
path-extremal.

Moreover, Lemma~\ref{lem:peeling-preserves-DHclass}(2) gives
$G-\{x,c\}\in \mathcal{D}(H)$, so the induction hypothesis implies that
$G-\{x,c\}$ is path-extremal. Thus both graphs on the right-hand side are
path-extremal.

Therefore, for every $k\ge 1$,
$z'(G;k)\ge z'(G-x;k)+z'(G-\{x,c\};k-1)
        \ge z'(P_{n-1};k)+z'(P_{n-2};k-1)
        = z'(P_n;k)$,
where the last equality follows from
Lemma~\ref{lem:path-nonforcing} and Pascal's identity. The case $k=0$ is
trivial, so $G$ is path-extremal.

\noindent\emph{Case 2: $T$ is a star centered at $P$.}
By Lemma~\ref{lem:prime-star-reduction}, either $G$ contains a twin pair, in
which case $G$ is path-extremal by Corollary~\ref{cor:twin-path-extremal}, or
there exists an induced subgraph $Q$ of $G$ isomorphic to the label graph of
$P$ such that $G$ is obtained from $Q$ by attaching pendant vertices to
vertices of $Q$. Now the same is true of $Q$ from the label  graph of $P$ being isomorphic to an induced subgraph of $H$. Consequently, every induced subgraph of $Q$ is path-extremal by the
hypothesis on $H$. Lemma~\ref{lem:pendant-extension} now implies that $G$ is
path-extremal.

This completes the induction. Since $\mathcal{C}(H)\subseteq \mathcal{D}(H)$,
the theorem follows. \qedbox
\end{proof}

\begin{corollary}\label{cor:bounded-prime-size}
Let $m\in \mathbb{N}$. Suppose that every induced subgraph of every split-prime
graph on at most $m$ vertices is path-extremal. Then every graph in
$\mathcal{C}_{\le m}$ is path-extremal.
\end{corollary}

\begin{proof}
Let $G\in \mathcal{C}_{\le m}$, and let $H$ denote the label graph of the
unique prime bag of the canonical split decomposition of $G$. By definition,
$|V(H)|\le m$. Hence every induced subgraph of $H$ is path-extremal by
hypothesis, so Theorem~\ref{thm:one-prime-bag} implies that $G$ is
path-extremal. \qedbox
\end{proof}

\begin{theorem}[Prime bags of order at most six]
\label{thm:small-prime-bags}
Every connected graph whose canonical split decomposition has all prime bags
of order at most six is path-extremal. More generally, let $H$ be a
split-prime graph such that every induced subgraph of $H$ is path-extremal.
Then every connected graph whose canonical split decomposition has all prime
bags of order at most six, apart from at most one prime bag with label
isomorphic to $H$, is path-extremal. Bag order includes marker vertices.
\end{theorem}

\begin{proof}
We first establish the universal bounds
\[
z(G;1)\le 2=z(P_n;1),\qquad
z(G;2)\le 3n-6=z(P_n;2)\qquad(n\ge3).
\]
A singleton can force only if $G$ is a path, in which case only its
endpoints can do so. For the pair bound, induct on $n$; the cases
$n=3,4$ are immediate. Let $n\ge5$. If $G$ has an isolated vertex,
every forcing pair contains it, giving at most $n-1$ pairs. If $x$ is
a leaf with neighbor $v$, Lemma~\ref{lem:leaf-recursion}, induction, and
the singleton bound give
\[
\begin{aligned}
z'(G;2)
&\ge z'(G-x;2)+z'(G-\{x,v\};1)\\
&\ge \binom{n-4}{2}+n-4
 =\binom{n-3}{2}=z'(P_n;2).
\end{aligned}
\]
Otherwise $\delta(G)\ge2$. Every forcing pair must make an initial
force, so it is an edge incident with a degree-two vertex. There are at most
$n$ such edges if every vertex has degree two, and at most $2(n-1)$
otherwise. Both quantities are at most $3n-6$.

It follows that every graph on at most six vertices is path-extremal:
the orders at most two are immediate, and for $n\le6$ and $k\ge3$,
Lemma~\ref{lem:path-nonforcing} gives $z'(P_n;k)=0$.
Not only that: every graph containing a nonempty fort $F$ of order at most
four is path-extremal. Indeed, we may assume $n\ge7$, and for $k\ge3$,
Lemma~\ref{lem:fort-obstruction} gives
\[
z'(G;k)\ge\binom{n-|F|}{k}\ge\binom{n-4}{k}
          \ge\binom{n-k-1}{k}=z'(P_n;k).
\]
The cases $k=1,2$ follow from the bounds established at the beginning of this theorem's proof\ (on $z(G;1)$ and $z(G;2)$), and $k=0$ is
immediate.

We now argue that every split-prime label $J$ is 2-connected. Indeed, if $v$ were a cut
vertex and some component $C$ of $J-v$ had at least two vertices,
then $(V(C),V(J)\setminus V(C))$ would be a split; if every component
were a singleton, $J$ would be a star. In particular, $\delta(J)\ge2$.
Now let $J$ label a prime leaf bag of order at most six, with marker
$p$. If $|V(J)|\le5$, then $V(J)\setminus\{p\}$ is a fort of
order at most four. If $|V(J)|=6$ and $p$ has a nonneighbor $q$,
then $V(J)\setminus\{p,q\}$ is a fort, since both excluded vertices
have at least two neighbors in it. If instead $p$ is universal, some
$q\ne p$ has degree at least three: otherwise $J-p$ would be a
1-regular graph on five vertices. Again $V(J)\setminus\{p,q\}$ is a
fort. These forts all consist of ordinary vertices: they avoid the marker. Now by accessibility a vertex outside of the bag sees either none of the fort or, exactly the neighbors of $p$ in it; so these are forts of $G$ as well.  This all implies that $G$ is gauranteed to be path-extremal if there is a prime leaf bag of order at most $6$.

We now prove both assertions by induction on $|V(G)|$. If there is at
most one prime bag, apply Theorems~\ref{thm:dh-path-extremal}
and~\ref{thm:one-prime-bag}; for a prime label of order at most six, the
required hereditary hypothesis follows from the preceding small-order
observation. Otherwise, one can root the bag tree at the exceptional bag labelled
by $H$, if present, and at any prime bag otherwise. Let $B$ be such that it is a leaf bag in the rooted subtree of a nonroot prime bag. If $B$ is prime,
it has order at most six and the fort construction gives us the desired statement; else  Lemma~\ref{lem:leaf-bag-geometry} and
Corollary~\ref{cor:multi-leaf-star-twins} either produce twins, which settle
the result by Corollary~\ref{cor:twin-path-extremal}, or show that $B$
consists of its marker, an ordinary center $c$, and a single ordinary
leaf $x$. In the latter case, the neighboring bag $A$ is not the root.
Let $a$ be its marker toward $B$.

As in the proof of Lemma~\ref{lem:peeling-preserves-DHclass}, deleting
$B$ and reinterpreting $a$ as the ordinary vertex $c$ represents
$G-x$; deleting $B$ and $a$ instead represents $G-\{x,c\}$.
If $A$ is non-prime, reducedness makes $a$ a clique vertex or a star
leaf, so its deletion leaves a connected degenerate label. If $A$ is
prime, its order is at most six, and deleting $a$ leaves a connected
label of order at most five. We can see what we want from considering what happens when we split-decompose this label and reduce the
resulting representation: a split decreases the order of each resulting
label, and in view of that subsequent suppression of order-two labels and degenerate-bag
joins neither enlarge prime labels nor change the exceptional label
$H$, so both smaller graphs remain connected and satisfy that all prime bags have order at most six except for at most one prime bag whose label is isomorphic to $H$.

It follows by induction that they are path-extremal. Writing $n=|V(G)|$,
Lemma~\ref{lem:leaf-recursion} and Pascal's identity therefore give, for
$1\le k\le n$,
\[
\begin{aligned}
z'(G;k)
&\ge z'(G-x;k)+z'(G-\{x,c\};k-1)\\
&\ge \binom{n-k-2}{k}+\binom{n-k-2}{k-1}
 =\binom{n-k-1}{k}=z'(P_n;k).
\end{aligned}
\]
The case $k=0$ is immediate. \qedbox
\end{proof}

\bibliographystyle{splncs04}
\bibliography{cocoon_zero_forcing_refs}

\end{document}